\documentclass[journal]{IEEEtran}
\ifCLASSINFOpdf
  \usepackage[pdftex]{graphicx}
\else
  \usepackage[dvips]{graphicx}
\fi
\usepackage{amsmath}
\usepackage{amsthm}

\usepackage{algorithmic}

\usepackage{array}

\ifCLASSOPTIONcompsoc
 \usepackage[caption=false,font=normalsize,labelfont=sf,textfont=sf]{subfig}
\else
 \usepackage[caption=false,font=footnotesize]{subfig}
\fi

\usepackage{stfloats}
\fnbelowfloat
\ifCLASSOPTIONcaptionsoff
 \usepackage[nomarkers]{endfloat}
\let\MYoriglatexcaption\caption
\renewcommand{\caption}[2][\relax]{\MYoriglatexcaption[#2]{#2}}
\fi
\usepackage{url}

\usepackage[numbers]{natbib}

\usepackage[switch]{lineno}

\usepackage{igor-commands-ieee}

\begin{document}

%
\title{Parks: A Doubly Infinite Family of NP-Complete Puzzles and Generalizations of A002464}
%
%
%


\author{Igor~Minevich,
        Gabe~Cunningham,
        K.~Aditya~Karan,
        and~Joshua~V.~Gyllinsky
\thanks{Minevich, Cunningham, \& Gyllinsky were with the School of Computing \& Data Science, Wentworth Institute of Technology, Boston,
MA, 02115 USA}
\thanks{Karan is with EAi Technologies, Vienna, VA.}
\thanks{Primary e-mail: minevichi@wit.edu.}
\thanks{Manuscript received November 2, 2024.}}

\maketitle

\begin{abstract}
The Parks Puzzle is a paper-and-pencil puzzle game that is classically played on a square grid with different colored regions (the parks). 
The player needs to place a certain number of ``trees'' in each row, column, and park such that none are adjacent, even diagonally. 
We define a doubly-infinite family of such puzzles, the $(c, r)$-tree Parks puzzles, where there need be $c$ trees per column and $r$ per row. We then prove that for each $c$ and $r$ the set of $(c, r)$-tree puzzles is NP-complete. 
For each $c$ and $r$, there is a sequence of possible board sizes $m \times n$, and the number of possible puzzle solutions for these board sizes is a doubly-infinite generalization of OEIS sequence A002464, which itself describes the case $c = r = 1$. 
This connects the Parks puzzle to chess-based puzzle problems,
as the sequence describes the number of ways to place non-attacking kings on a chessboard so that there is exactly one in each column and row (i.e. to place non-attacking dragon kings in shogi).
These findings add yet another puzzle to the set of chess puzzles and expands the list of known NP-complete problems described.
\end{abstract}


%
\IEEEpeerreviewmaketitle

\section{Introduction}
\label{sec:Intro}


The P versus NP problem is a major open problem in computer science, and one of the seven Millennium Problems posed by the Clay Mathematical Institute in the year 2000, with a \$1,000,000 prize for a solution. As of this writing, only one of the seven Millennium Problems has been solved \cite{devlin2002millennium, carlson2006millennium, PoincareConj}.
Despite decades of research into computational complexity, the current state of the theoretical understanding leaves
much to be desired. 
For a detailed exposition on the P vs. NP problem itself, as well as its place in computer science research, see \cite{gary, sipser, 10.1145/602382.602398}.

The problem of determining the complexity class of a computational problem is an important aspect of the study of complexity classes, of which P and NP have taken center stage due to the practical outcomes associated with their study. 
The study of the P and NP complexity classes is a very mature field. 
The study of puzzles in NP is also quite mature, with several decades worth of work related to identifying NP-complete puzzles, 
developing strategies for proving reductions, and advancing algorithmic techniques to create practical solvers for these puzzles \cite{ComplexityOfPuzzles, Hearn2009games}. 
A large number of common
puzzles such as Minesweeper \cite{Minesweeper2003}, Sudoku \cite{Kendall2008}, Kakuro \cite{Kakuro} and Yosenabe \cite{Yosenabe} have been proven NP-complete. A relatively large number of lesser known puzzles have also been shown to be NP-complete, such as
Light Up, Numberlink, and KPlumber 
\cite{KRAL2004473}. In 2012, Andrea Sabbatini released a cell phone app called ``100 Logic Games - Time Killers'' that contains 100 different types of puzzles, most of which have been proven to be NP-complete already: Tents \cite{TentsNPC}, Nurikabe \cite{NurikabeNPC}, Skyscrapers \cite{SkyscrapersNPC}, Battleships \cite{BattleshipsNPC}, Hitori, and Kakuro just to name a few; see \cite{ComplexityOfPuzzles} for a list of many such and more. Yet the authors are not aware of a rigorous proof of the NP-completeness of the Parks puzzle, at least not in its full generality (the senior thesis of the third author contains a proof in a particular case \cite{karan2020parks}).

The Parks Puzzle is a logic puzzle akin to Sudoku that has not to our knowledge been rigorously shown to be NP-complete, at least not in its full generalization. This puzzle challenges players to determine the placement of a certain number of trees in a grid, adhering to simple rules: place exactly $c$ trees in each column, exactly $r$ in each row, and exactly $r$ in each park, making sure no two trees are adjacent, not even diagonally (see Figures \ref{fig:1TreeParksEg}, \ref{fig:2TreeParksEg}, \ref{fig:21TreeParksEg}, and \ref{fig:12TreeParksEg}; the ``parks'' are recognized by their own distinct colors). This puzzle is growing in popularity via several mobile apps made by Andrea Sabbatini \cite{ParksSeasons, LogicGames, ParksCantica, ParksLandscapes}, yet the analysis of its difficulty and fascinating underlying combinatorics remains unstudied. 

We start by defining the Parks puzzles and giving a few examples in Section \ref{sec:Rules}. Notice that in so doing we are presenting a more general version of the puzzle than those in the aforementioned apps\footnote{The mobile apps listed above generalize the puzzle a bit less, with $(t, t)$-tree puzzles where we must place $t$ trees in each column, row, and park.}. This generalization presents new possibilities for creating interesting puzzles, and has wider mathematical and computer science implications. 

Then we review a bit of background on the P vs. NP Millennium Prize Problem and the 3-SAT logic problem; for a thorough understanding of the first problem, we refer the reader to \cite{gary} or \cite{sipser}. Next, we prove Parks is in NP (Section \ref{sec:ParksNP}), present a proof of the NP-completeness of the $(1, 1)$-tree puzzle in Section \ref{sec:11NPC}, and finally a proof of the NP-completeness of the general $(c, r)$-tree puzzle in Section \ref{sec:crNPC}.
The proofs of Sections \ref{sec:11NPC} and \ref{sec:crNPC} are reductions from 3-SAT.
This shows that while the puzzle is easy to describe and every puzzle in the aforementioned mobile apps is solvable without the need for guessing, since there is a unique solution (as claimed in the apps), there is probably\footnote{This is assuming $P \ne NP$.} no algorithm possible for a classical computer that would solve an arbitrary-sized puzzle without essentially some amount of brute-force guessing.

Next, we dive into the combinatorics underlying the puzzle, namely the number of puzzles of a given size and the number of possible solutions, primarily focusing on the latter. 
It turns out that the smallest size $(c, r)$-tree puzzle is $4c \times 4r$ (Theorem \ref{thm:smallest_size}), and there are exactly 2 tree configurations possible for that size. 
Thus, for the construction of non-trivial puzzles it would seem prudent to choose larger puzzles. 
For any pair of $c$ and $r$ there is a sequence of relevant board sizes and the number of $(c, r)$-tree configurations for these board sizes could be a new OEIS sequence, except for the case of $c=r=1$, which is OEIS sequence A002464 and counts the number of tree arrangements for the basic $(1, 1)$-tree puzzle. Thus we relate the ever-developing study of complexity theory to combinatorics in a new way and open the door to interesting combinatorial questions related to a doubly-infinite family of NP-complete puzzles.

%
%
%

\section{Introducing the $(c, r)$-tree Parks Puzzle}
\label{sec:Rules}
\begin{defn}
\label{defn:Parks}
    Let $c$ and $r$ be positive integers. We define a $(c, r)$-tree Parks puzzle to be an $m \times n$ board with $m$ regions called parks (hence the name of the puzzle), each marked with a different color on the grid. Typically, the parks are arranged in a way that the squares of any one park are contiguous (i.e. form a connected subgraph of the grid). The goal of the puzzle is to place trees in the squares of the grid (at most one in any square), such that:
\begin{itemize}
    \item each column contains exactly $c$ trees,
    \item each row contains exactly $r$ trees,
    \item each park contains exactly $r$ trees\footnote{It is also possible to consider specialized Parks puzzle parameters, such as alternatives to $m$ parks and $r$ trees in each, but then it is not clear if NP-completeness still holds.}, and
    \item no two trees are on squares that are adjacent to one another (not even diagonally).
\end{itemize} 
\end{defn}

We could also define a version of the Parks puzzle where instead we have $n$ parks, with the goal of placing $c$ trees in each, but that is the same as rotating the puzzle by $90^\circ$.

One immediately notes that 
\begin{equation}
    \label{eqn:size_limit}
    cn = rm
\end{equation} trees need to be placed in the grid, so we are limited to sizes $m \times n$ such that \eqref{eqn:size_limit} holds.

The most basic version of Parks puzzles is the $(1, 1)$-tree version, where we must place 1 tree in each column, each row, and each park. Andrea Sabbatini's mobile phone apps Logic Games \cite{LogicGames}, Parks Cantica \cite{ParksCantica}, and Parks Seasons \cite{ParksSeasons} 
have this version of the puzzle as well as the $(2, 2)$- and $(3, 3)$-tree versions, all of which are limited to square grids by \eqref{eqn:size_limit}. We will refer to $(t, t)$-tree puzzles simply as $t$-tree puzzles. 

Here is how one might approach solving an example 1-tree $6 \times 6$ Parks puzzle (see Figure \ref{fig:1TreeParksEg}). In this example the starting position is shown on the left and its solution on the right, where a ``T'' denotes a tree and an ``X'' denotes a square that the solver has determined to definitely \textit{not} have a tree. The reader will readily notice the yellow park has only one square, namely F2, so that square must have a tree; from there, we can cross out the squares diagonally adjacent to F2 and the ones in the same row or same column. 
The first row is all red (i.e. is entirely contained in the red park), so because there must be a tree in the first row, the red park must have its only tree in the first row; similarly, since Column A is all red, the red tree must be in Column A, hence the red tree is in A1, and we can cross out the rest of the red squares. Because the cyan park has only the two squares E5 and E6, we cannot have trees in E4, D5, or D6 (a tree in any of these places would prevent the cyan park from having a tree). Now, a tree in D4 would prevent Row 5 from having a tree so the only square left in the green park is in D3, therefore D3 must have a tree, and from there it follows that B4, next C6, and finally E5 must have trees.

\begin{figure}[hbtp]
    \centering
    \includegraphics[height=1in]{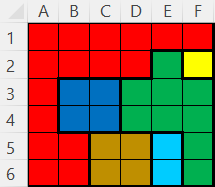} \hspace{6pt} \includegraphics[height=1in]{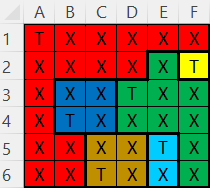}
    \caption{A 1-tree $6 \times 6$ Parks puzzle.}
    \label{fig:1TreeParksEg}
\end{figure}

A computer can readily be programmed to take into account logic similar to that demonstrated above to quickly solve the subset of puzzles that are amenable to solution by such simple logic. However, as we will prove, the general problem of determining whether or not a Parks puzzle has a solution is NP-complete. Thus, ultimately there is probably no way to program a computer to quickly solve larger, more complex puzzles\footnote{Again, assuming P $\ne$ NP.}, no matter how complex the solver. 

Figure \ref{fig:2TreeParksEg} below shows an example of a 2-tree Parks puzzle (left) and its solution (right); we leave it as an exercise to the reader to show that the solution presented on the right in the figure is the only one. 
As the reader may discover, the puzzle in Figure \ref{fig:2TreeParksEg} requires slightly more complicated logic to solve than the puzzle in Figure \ref{fig:1TreeParksEg}, and it would seem that in general, the logic behind solving $t$-tree puzzles becomes more difficult for a human solver to comprehend as $k$ grows larger\footnote{At the time of writing the versions of the Logic Games, Parks Landscapes, and Parks Seasons apps appear to only have puzzles with $k \le 2$, and the Parks Cantica app has puzzles with $k \le 3$, with only a small portion of those puzzles having $k = 3$.}.

\begin{figure}[hbtp]
    \centering
    \includegraphics[height=1in]{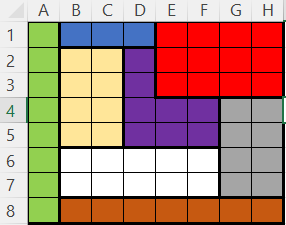} \hspace{6pt} \includegraphics[height=1in]{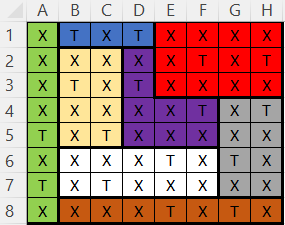}
    \caption{A 2-tree $8 \times 8$ Parks puzzle.}
    \label{fig:2TreeParksEg}
\end{figure}

Non-square $(c,r)$-tree Parks puzzles may offer a greater variety of puzzles, since their complexity lies in-between $t$-tree and $(t+1)$-tree Parks puzzles, yet the authors were unable to find examples of non-square Parks puzzles. 
In Figure \ref{fig:21TreeParksEg} we present a (2, 1)-tree puzzle (on the left) and its solution, which we leave to the reader to verify as the unique one. Likewise, in Figure \ref{fig:12TreeParksEg} we present a (1, 2)-tree puzzle (on top) and its solution (below), which we also leave to the reader to verify as unique. 

\begin{figure}[hbtp]
    \centering
    \includegraphics[height=1in]{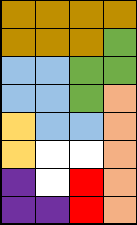} \hspace{6pt} \includegraphics[height=1in]{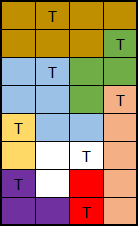}
    \caption{A (2, 1)-tree $8 \times 4$ Parks puzzle (left) and its solution (right).}
    \label{fig:21TreeParksEg}
\end{figure}

\begin{figure}[hbtp]
    \centering
    \includegraphics[height=0.5in]{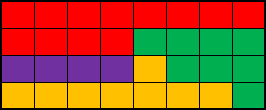}\vspace{12pt}\\
    \includegraphics[height=0.5in]{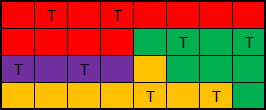}
    \caption{A (1, 2)-tree $4 \times 8$ Parks puzzle (top) and its solution (bottom).}
    \label{fig:12TreeParksEg}
\end{figure}

\section{P vs. NP}

Recall that a decision problem is a computation problem whose answer must be ``yes'' or ``no''. The general Parks Puzzle decision problem is to determine, given
a board with some marked trees, if there is some configuration of
the board that is a valid solution to the puzzle. In our paper we
show that the subset of the decision problem without any marked
trees, Parks, is NP-complete, proving a lower bound for the general
decision problem. Recall that an ``oracle'' is an abstract machine that resembles a ``black box'' that is able to give an answer to any decision problem; a polynomial time oracle gives that answer in polynomial time (polynomial in the length of the input string). Given a polynomial time oracle for the Parks decision problem, the ``function'' problem of finding an explicit solution to a Parks puzzle would be trivially solvable in polynomial time by
checking every square (placing a tree in it and asking the oracle if there is still a valid solution), each in polynomial time. Thus, the function problem is essentially just as difficult as the decision problem: if one can be solved in polynomial time, so can the other. Therefore, in this paper we focus on the decision problem and prove that the Parks Puzzle decision problem is NP-complete, which suggests that it is very
unlikely that there exists any efficient (polynomial time) algorithm to either check that a puzzle has a solution or to find the solution.

The first problem proven to be NP-complete was SAT \cite{johnson2012brief, Cook1971}, the Boolean Satisfiability Problem: given a set of Boolean variables $x_1, x_2, \dots, x_n$ (which can be set to True or False), and a Boolean expression such as $(x_1 \vee x_2) \wedge (\neg x_2 \vee x_3)$, where ``$\vee$'' is the logical ``OR'', $\wedge$ means ``AND'', and $\neg$ means ``NOT'' (the negation of the variable's value), the problem is to determine whether or not there is a way to assign values to the variables such that the logical expression can be satisfied, i.e. evaluates to True. 
At the same time, a close relative of SAT called 3SAT was proven NP-complete \cite{Cook1971}, and it is this problem we will use to show that the Parks Puzzle is NP-complete. 3SAT, also known as ``3-Satisfiability'', is a subset of the SAT problem that is restricted to logical expressions in \textbf{3CNF} (3-conjunctive normal form). 
These expressions consist of an arbitrary number of 3-disjunctive clauses such as $(x_1 \vee \neg x_2 \vee x_3)$ that involve three variables or their negations in a \textit{3-disjunction} (a three-input OR gate). These clauses are then all put in a conjunction (an AND gate with an arbitrary number of inputs). Some examples of 3CNF expressions are: $x \vee y \vee z$, $(x \vee \neg x \vee y) \wedge (y \vee z \vee \neg w)$, and $(x \vee \neg y \vee \neg z) \wedge (\neg x \vee y \vee w)$.

\section{Parks is in NP}
\label{sec:ParksNP}

An instance of a $(c, r)$-tree $m \times n$ Parks puzzle $\pi$ can be inputted into a Turing Machine, roughly speaking, as the four numbers $c, r, m, n$ followed by a list of color numbers corresponding to each coordinate in the grid in order, starting with the top row, going from left to right, then next row, etc.
The set of $(c, r)$-tree $m \times n$ Parks Puzzles with a valid solution is the language of interest. 
A certificate $\sigma$ for a $(c, r)$-tree Parks puzzle is a list of coordinates marking the locations of the trees.
The $(c, r)$-tree Parks verification problem is: given $\langle\pi, \sigma\rangle$, where $\pi$ is a Parks puzzle and $\sigma$ a possible certificate for $\pi$, does $\sigma$ correspond to a valid solution of the puzzle, satisfying the rules in Definition \ref{defn:Parks}?
In other words, do the trees in $\sigma$ correspond to a valid solution of $\pi$? Moreover, can we check whether or not this is true in polynomial time in the size of the puzzle $\pi$?

It is easy to see that $(c, r)$-tree Parks is in NP without any elaborate descriptions or algorithms, but we will provide one such algorithm here for completeness. 
The data for an $m \times n$ Parks puzzle can be stored using $N = O(mn\log m)$ symbols, since there are $m$ colors and each can be represented with a number using $\log m$ digits.
Given the puzzle $\pi$ and certificate $\sigma$, $\sigma$ is obviously of polynomial size in $N$ (the number of trees is no more than the number of squares in the puzzle), and the validity of $\sigma$ can be checked as follows.
First, we sort $\sigma$ in dictionary ordering, by row and then column (which takes $O(N \log N)$ time \cite{CLRS}) and check that each row has exactly $r$ trees in $O(mr) = O(N)$ 
time. 
Then we check that no trees are adjacent by going through this sorted list and for each pair of adjacent rows, sort the column numbers of the trees in that pair of rows, using the algorithm for Merge Sort for example, and ensure that the horizontal distance between any two adjacent items is at least 2, which takes $O(mr) = O(N)$ time. 

Next, we check that there are $c$ trees in each column and $r$ trees in each park as follows. First, run through the input $\pi$ and store the color numbers in an $m \times n$ array. Then we go through $\sigma$, adding to a counter for each column and each park every time we find a tree belongs there. Finally, we make sure the column counters all equal $c$ and park counters equal $r$. This process takes $O(N)$ time total. The total time is $O(N \log N)$, so $\langle \pi, \sigma \rangle$ can be verified in polynomial time and $(c, r)$-tree Parks is in NP. 
Note that the proof assumes nothing about the contiguity of the parks.



\section{1-Tree Parks is NP-Complete}
\label{sec:11NPC}

We have already shown that for any $c$ and $r$, the $(c, r)$-Parks is in NP, so it remains to show the decision problem is NP-hard, which we will do via a polynomial-time transformation from 3SAT. As stated, the decision problem says nothing about the parks being contiguous regions, but since this makes more sense to a human, we will use contiguous parks everywhere, deriving the stronger result that the subset of $(c, r)$-tree Parks puzzles $\pi$ where the parks are contiguous is itself enough to be NP-hard. We will first show that 1-tree Parks is NP-hard, then in the next section generalize these ideas to $(c, r)$-tree puzzles.

Given a 3CNF expression $\phi$, we will build a corresponding Parks puzzle $\pi$ by first identifying all the variables involved, and putting in ``variable parks'' that correspond to these (see Figure \ref{fig:var_park}); for the 1-tree puzzle, these will be $1 \times 2$ parks, where a tree in the left square (A1) will correspond to the variable in $\phi$ being set to True, and a tree in the right square (B1) will correspond to the the variable being False. We will build the Parks puzzle in a way that any assignment of values to the variables in $\phi$ will make $\phi$ evaluate to True if and only if placing the corresponding trees in the variable parks in $\pi$ gives a valid solution to the puzzle $\pi$.

\begin{figure}[hbtp]
    \centering
    \includegraphics{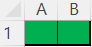}
    \caption{A variable park in a 1-tree puzzle}
    \label{fig:var_park}
\end{figure}

We will build the puzzle $\pi$ by making IFF ``gadgets'' (gadgets are pieces of the puzzle used for certain specialized functionality) that allow us to access and reuse the values of these variables as much as necessary. The columns corresponding to the True / False values of variables will then be used to build OR gadgets that correspond to each of the 3-disjunctive clauses. The key observation that makes all this work is the following: the 3CNF expression $\phi$ has a solution if and only if each 3-disjunctive clause can be satisfied with some assignment of variable values, which is true if and only if each OR gadget in the puzzle can be filled, and this happens if and only if the whole puzzle has a solution. And because of the way we make variable parks, there is a one-to-one correspondence between the solutions to the puzzle and the assignments of variable values that make $\phi$ evaluate to True. Just as the tree in the red park in Figure \ref{fig:1TreeParksEg} must be in the upper-left because the whole left column and top row are red, we will make use of this reasoning and make a white top row and right column, so the tree in the white park will necessarily be in the upper-right, and then the remaining white squares can be used as empty space and to separate parks. We will also use 1-square parks to cancel out a row or a column as needed. 

In the following examples, we will mark basic facts about where there must be trees with a ``T'', or must not be trees with an ``X'', just as we have done in Figure \ref{fig:1TreeParksEg}, so that the relationship between our various ``gadgets'', and parks within those gadgets, becomes more evident.

Figure \ref{fig:IFF_1tree} shows an example IFF gadget. Note that the three green parks, which correspond to variables, must have the same ``value'', i.e. either they all have a tree on the left (as in the top part of the figure) or they all have a tree on the right (as in the bottom). The reader should verify that in both parts of the figure, if we place any one tree that is pictured, the remaining trees in the same picture would have to be placed. Therefore, we may rightly call all of the green parks by the same variable name, and use their columns whenever we have that variable come up in a 3-disjunction.

\begin{figure}[hbtp]
    \centering
    \includegraphics[height=2in]{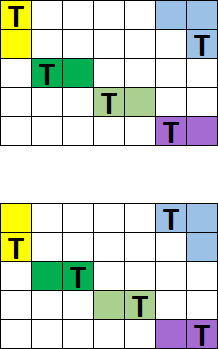}
    \caption{IFF gadget in a 1-tree puzzle. Top: all variables are set to True; bottom: all are set to false.}
    \label{fig:IFF_1tree}
\end{figure}

Figure \ref{fig:OR_1tree} shows an example of an OR gadget, which creates a disjunction between the three variables corresponding to the two-square parks in the bottom three rows, which are all variable parks. As such, the gate itself only consists of the top 9 rows of Figure \ref{fig:OR_1tree}. The variable parks may or may not be inside an IFF gate, and they may even be inside the same one. Note the gray-colored parks are all separate parks and should technically each have their own color; these are not necessarily directly adjacent to the main part of the gate (the yellow, brown, and red parks), but rather are collected to the right of the IFF gates and variables, each in its own column, as shown in Figure \ref{fig:BigEg}. Rows 1, 4, and 7 make up the actual functionality of the gadget, while rows 2, 5, and 8 only serve to connect the pieces of the parks in rows 1, 4, and 7 so as to have contiguous parks. Columns C and F are only there to show that there may be a separation between the variable parks used. The separation may not even be there if two of the variable parks are in the same IFF gate, or the separation may in actuality be quite large, and the gadget would work just as well; the only adjustment necessary is to fill in rows 2, 5, and 8 with yellow, brown, and red squares, respectively, to just make sure those parks are contiguous.

\begin{figure}[hbtp]
    \centering
    \includegraphics[height=1in]{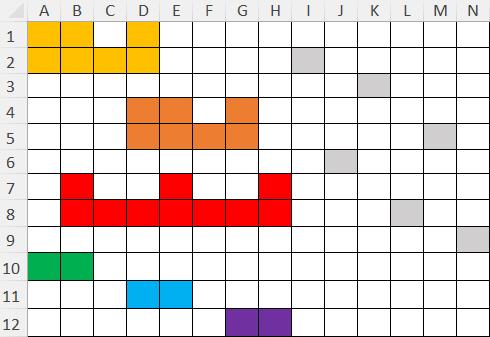}
    \caption{OR gadget in a 1-tree puzzle}
    \label{fig:OR_1tree}
\end{figure}

The OR gadget in Figure \ref{fig:OR_1tree} works as follows. The red park ensures that if all three variables are set to False (so the trees in all three variable parks are on the right), then there is no solution to the puzzle because there would be nowhere to put a tree in the red park. The yellow and brown parks then make sure there is a unique valid solution to the puzzle in all other cases, i.e. if at least one variable is set to true. The reader should verify that there is exactly one way to place trees in the yellow, brown, and red parks if at least one of the variables is set to true (there are 7 cases to check). The yellow, brown, and red parks ensure that there is a tree in each column where the trees in the variable parks are missing.

The OR gadget in Figure \ref{fig:OR_1tree} can be easily modified to account for the negation of a variable instead by swapping the two columns corresponding to that variable, readjusting rows 2, 5, 8 as necessary to ensure contiguity. For example, if we needed a disjunction involving the negation of the variable corresponding to the purple park, we would switch columns G and H, add a brown square in G5, and remove the red square in H8. Note that because all the pieces (yellow, brown, red and the variable parks) are essentially disjoint from one another, the adjacency rule for trees plays no effect here. That means switching columns affects the set of solutions to the puzzle in precisely the way that we want: every possible solution just has the trees in those columns switched.

All that remains is to put these gadgets together, like the pieces of a huge puzzle. Figure \ref{fig:BigEg} shows the puzzle $\pi$ constructed from the 3CNF expression 
\[\phi = (X \vee Y \vee Z) \wedge (\neg X \vee Y \vee W).\]
The construction starts by making an IFF gadget for every variable that appears more than once in the expression $\phi$, creating as many variable parks inside the IFF gadget as the number of times the variable or its negation appears (the two green parks for the $X$ variable and the two purple parks for $Y$ in this case); we have written an uppercase ``X'' below the squares which, if they contained a tree, would correspond to a True value for $X$ and a lowercase ``x'' under those which would correspond to a False value for $X$, and similarly for the other variables. We then create variable parks for those variables that appear only once, outside of an IFF gadget (it is not needed for the 1-tree puzzle but will be needed for the $t$-tree puzzle in general in the next section), and we separate the variable parks and IFF gadgets so they are not in adjacent columns by adding the gray one-square parks below. Next, we create the OR gadgets above (or they could just as well be placed below) the variable parks, making sure that if a variable appears in a negative form (namely, $\neg X$ in the second 3-disjunction), then we switch those corresponding columns of the OR gadget. Finally, we make sure the white region has a tree in the upper-right by making the entire rightmost column and top row white.


\begin{figure}[hbtp]
    \centering
    \includegraphics[width=\columnwidth]{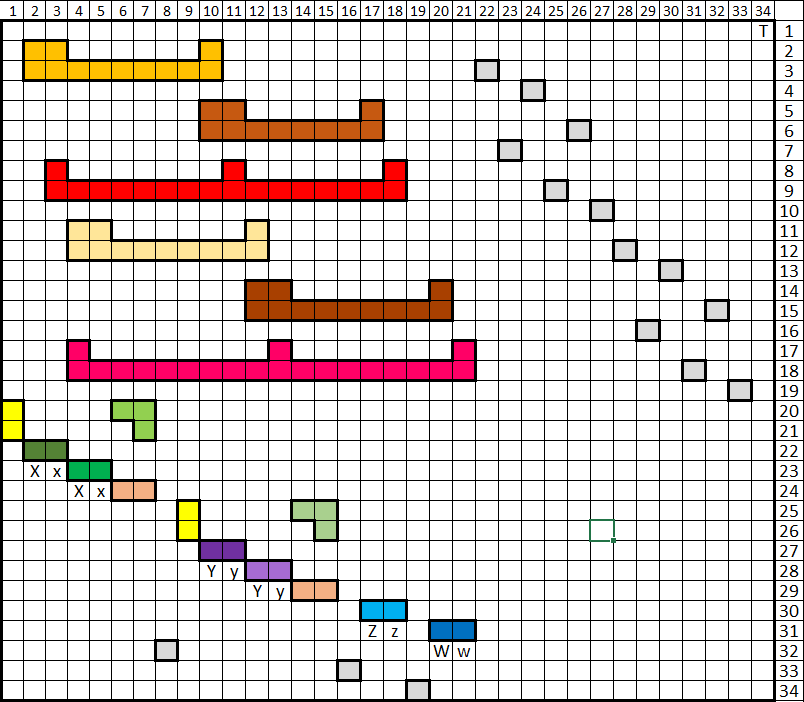}
    \caption{Example of the puzzle $\pi$ constructed from the 3CNF expression $(X \vee Y \vee Z) \wedge (\neg X \vee Y \vee W)$}
    \label{fig:BigEg}
\end{figure}

\section{$(c, r)$-Tree Parks is NP-Complete}
\label{sec:crNPC}
The previous section proved the NP-completeness of 1-tree Parks, and in this section we extend these ideas to $(c, r)$-tree Parks. The idea is to sort of stretch the gates used in 1-tree Parks $r$ times horizontally, and to supplement with $1 \times (2r-1)$-sized parks vertically to account for the number of trees needed in each column. 

\subsection{Variable parks}
Thus a ``variable park'' for a boolean variable $x$ will be a $1 \times 2r$-sized park, with only two possible configurations in it, as shown in Figure \ref{fig:crVarPark} in the case $r = 3$. The left configuration of trees represents $x$ being true, and the right configuration - $x$ being false. Any variable park will be enclosed in an IFF gate (see Figure \ref{fig:crIFF}), which is what ensures that those are the only two possible tree configurations.

\begin{figure}[hbtp]
    \centering
    \includegraphics[height=0.25in]{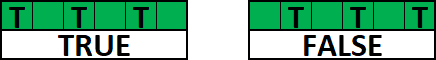}
    \caption{$(c, 3)$-tree variable park}
    \label{fig:crVarPark}
\end{figure}

\subsection{IFF gate}
We next present the IFF gate, an example of which is shown for $(2, 3)$-tree Parks in Figure \ref{fig:crIFF}. The trees that must always be placed in their respective places, no matter what the values of the variables, are already written in as T's in Figure \ref{fig:crIFF}. Here, the two shades of green represent two variable parks whose values are equivalent; this could be extended to the right in the obvious way to account for any number of variable parks that would all be equivalent, or we could simply remove Columns M - R if we only need to use the variable park ones, shifting the brick-colored parks in rows 4, 5, 7 down one and those in rows 1 and 2 down by two rows. The yellow park consists of a single column of two yellow squares together with $2r-2$ squares extending to the left of the bottom square. Note that, since this park must contain $r$ squares, we must have a tree in the park's rightmost column (Column F in Figure \ref{fig:crIFF}) and therefore there is no tree in the next square over (E10 in Figure \ref{fig:crIFF}) and then the trees must be placed in every other square from then on. The purple park is just a $1 \times 2r$-sized park. The light blue park consists of the same turned L as we had for 1-tree Parks with $2r-2$ more squares to the left of the top row, as in Figure \ref{fig:crIFF}. 

The brick-colored $1 \times (2r-1)$ parks in rows 1 - 8 and in the bottom-left of Figure \ref{fig:crIFF} are meant to all be distinct parks, so they technically should have different colors; they are only the same color here for convenience of discussion here. These parks vanish if $c = 1$, and if $c > 2$ then these rows are replicated above the current picture $c-2$ more times. Note that, except for the lower-left brick-colored park, this configuration reduces to the 1-tree IFF gate in Figure \ref{fig:IFF_1tree}.

\begin{figure}[hbtp]
    \centering
    \includegraphics[width=\columnwidth]{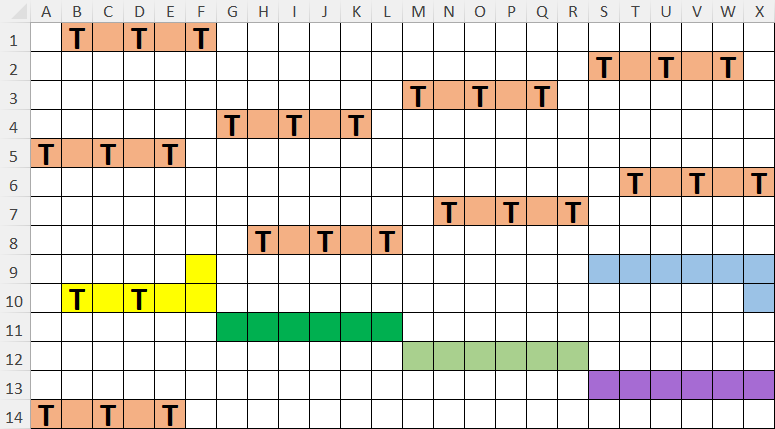}
    \caption{$(2, 3)$-tree IFF gate}
    \label{fig:crIFF}
\end{figure}

The reader should verify that there are only two possible tree configurations for the IFF gate:
    \begin{enumerate}
        \item There is a tree in the uppermost square of the yellow park and all the variable parks are set to true as in Figure \ref{fig:crVarPark}.
        \item There is a tree in the lower-right square of the yellow park and all the variable parks are set to false as in Figure \ref{fig:crVarPark}.
    \end{enumerate}
Regardless of the configuration of the IFF gate, each column outside of the columns of the variable parks, as well as those columns of the variable parks that contain trees, will have exactly $c$ trees. The columns where the variable parks do not have a tree will have $c-1$ trees. Finally, each row and each park will have exactly $r$ trees.

\subsection{OR gate}
The last gadget is the OR gate, which is simply the stretched out version of the OR gate in Figure \ref{fig:OR_1tree}, as presented in Figure \ref{fig:crOR} in the case $r = 2$. The three variable parks here (green, blue, and purple) will be part of IFF gates and, per the discussion above, there will be $c-1$ trees in each column where the trees in the variable parks are missing, so the yellow, brown, and red parks serve to add exactly one tree in each such column.

\begin{figure}[hbtp]
    \centering
    \includegraphics[width=\columnwidth]{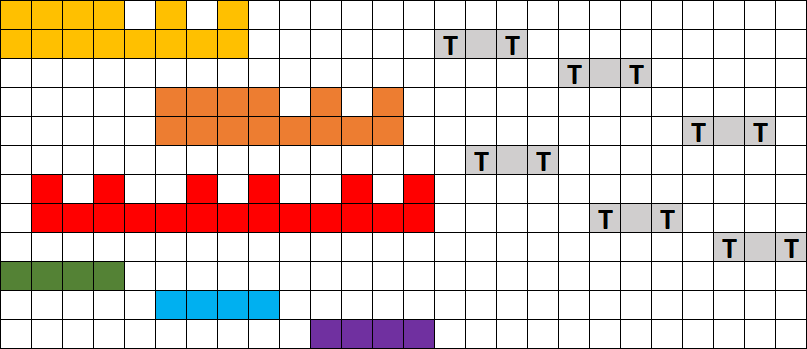}
    \caption{$(c, 2)$-tree OR gate}
    \label{fig:crOR}
\end{figure}

If $c > 1$, we will need to put in more parks like the gray ones on the right of Figure \ref{fig:crOR} below the gray ones currently there, in the same order of rows, using $6(c-1)$ additional rows to do so (if additional OR gates are used, these parks could go in the corresponding rows of those gates instead of new rows to save space, as we have done in Figure \ref{fig:2BigEg}). As these already have $r$ trees in them, after doing so we will have exactly $c$ trees in each column and $r$ trees in each row and each park corresponding to the OR gate. If we need to use the negation of the boolean value, we again simply switch each corresponding pair of columns as in the 1-tree Parks discussion. For example, if we need to use the negation of the green variable, we would switch the first two columns, as well as the 3rd and 4th columns, in Figure \ref{fig:crOR}.

\subsection{Putting it all together}
Note that the IFF gates in Figure \ref{fig:crIFF} can be placed diagonally next to each other without interfering with each other. The top row of the entire puzzle will be white (the same color as the white space in all of the gadgets shown in Figures \ref{fig:crIFF} and \ref{fig:crOR}) and the very last row of the rightmost OR gate will be missing the gray park, so that the trees that would have gone there must go into the white park instead (in the top row). See Figure \ref{fig:2BigEg} for a $68 \times 68$ puzzle that is the 2-tree version of Figure \ref{fig:BigEg}. Note that here, we do not need the one-square parks in the bottom 3 rows of Figure \ref{fig:BigEg} to separate the IFF gates, but we do need to make much more room for the IFF gates that hold the Z and W variables.

\begin{figure}[hbtp]
    \centering
    \includegraphics[width=\columnwidth]{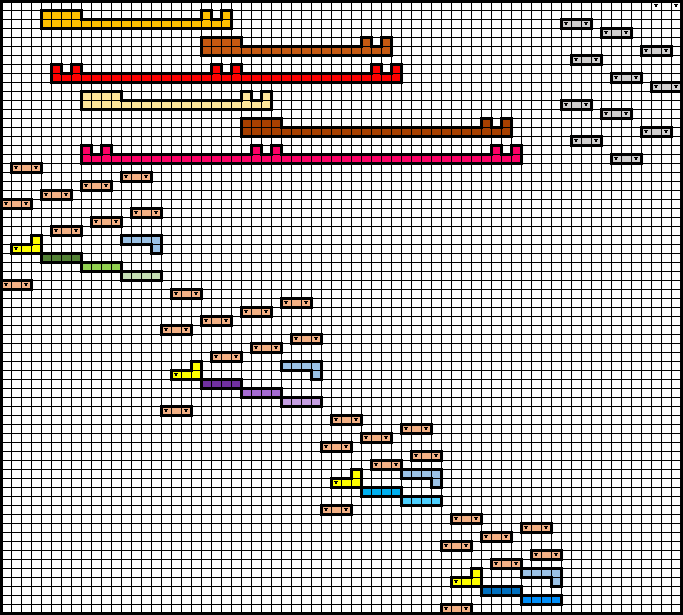}
    \caption{Example of a $68 \times 68$ 2-tree puzzle $\pi$ constructed from the 3CNF expression $(X \vee Y \vee Z) \wedge (\neg X \vee Y \vee W)$}
    \label{fig:2BigEg}
\end{figure}

The gate constructions and preceding discussion ensure that the resulting $m \times n$ $(c, r)$-tree puzzle satisfies \eqref{eqn:size_limit} and will have a unique solution for any configuration of trees that corresponds to a valid substitution of values into the 3CNF expression from which it was built.

In constructing this puzzle, if we have $M$ boolean variables $x_1, x_2$, $\dots, x_M$, and if we have $N$ 3-disjunctions, with each variable $x_i$ appearing $a_i$ times among all disjunctions, then we will need to use no more than:
\begin{itemize}
    \item $M$ IFF gates, where the one for variable $x_i$ comes with $4r + 2ra_i$ columns and $8(c-1) + 4 + a_i$ rows and
    \item $N$ OR gates, each of which comes with at most $3 + 6c$ new rows and $6r$ new columns.
\end{itemize}

Overall, the total number of rows and columns is a polynomial $p(c,r,M,N)$ so this is a polynomial-time reduction from 3CNF. This proves:

\begin{thm}
For any positive integers $c$ and $r$, the family of $(c, r)$-tree Parks puzzles is NP-complete.
\end{thm}

\section{Minimal Configurations of Trees}
\label{sec:Shuriken}
Now that we know the general family of $(c, r)$-tree Parks puzzles is NP-complete, we have reason to believe there is no general algorithm for solving an arbitrary-sized puzzle without doing some amount of guesswork in general. For smaller puzzles, on the other hand, it may be possible to write an algorithm that will find the solution without any guessing. In particular, we will now show that the smallest size nontrivial $(c, r)$-tree puzzle is $4c \times 4r$, and there are only two tree configurations possible for this size. These configurations have no tree placements in common, so given a puzzle of this size, if there is even a single park that contains no squares from one of the two configurations but only contains squares from the other configuration, then we can immediately tell what the solution should be. Or, if we can find a single tree or even cross out a square containing a tree from one of the configurations, that immediately tells us the solution as well.  

\begin{thm}
    \label{thm:smallest_size}
    Besides the trivial $1 \times 1$ puzzle for $c = r = 1$, the smallest possible size of grid that can have a valid configuration of trees for a $(c, r)$-tree puzzle is $4c \times 4r$. Furthermore, there are only 2 possible configurations of trees for the $4c \times 4r$ puzzle, namely the ``Shuriken arrangements'' described below.
\end{thm}
\begin{proof}
Suppose we have a $(c, r)$-tree Parks puzzle that has more than one row. Consider any two adjacent rows of the puzzle, and break them up into $2 \times 2$ boxes as shown in Figures \ref{fig:8x8config} and \ref{fig:8x12config}, with exception of possibly the rightmost box, which will be $2 \times 1$ if the number of columns is odd. Since a tree cannot be adjacent to another tree, even diagonally, each box contains at most one tree. These two rows need to have $2r$ trees in them, so by the pigeonhole principle we need at least $2r$ boxes, which implies the number of columns is at least $4r-1$. Now the only way to have only $4r-1$ columns is for \textit{every} pair of adjacent rows to have trees in the odd-numbered columns only, which is impossible since the even-numbered columns need trees as well. So there must be at least $4r$ columns and, by similar reasoning, at least $4c$ rows. 

To make the rest of the proof easier to understand, let us label the four ``types'' of $2\times 2$ boxes by where the tree is positioned in each box, as in Figure \ref{fig:4boxes}. 

\begin{figure}[hbtp]
    \centering
    \includegraphics[height=.55in]{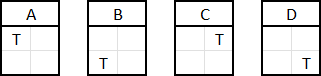}
    \caption{The 4 types of $2 \times 2$ boxes (labelled at the top by type).}
    \label{fig:4boxes}
\end{figure}

First we present a description of the two valid tree positions for $(c,r)$-tree $4c \times 4r$ puzzles. We call these positions ``Shuriken arrangements'' because they sort of resemble the shape of a Japanese Shuriken, whose blades can be extended either vertically or horizontally. Figure \ref{fig:1shuriken} shows the case $c = r = 1$, and Figure \ref{fig:12shuriken} shows the cases $c = 1, r = 2$.
    
    \begin{figure}[hbtp]
        \centering
        \includegraphics[height=0.5in]{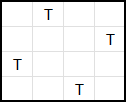} \hspace{6pt} \includegraphics[height=0.5in]{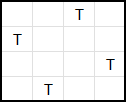}
        \caption{The Shuriken arrangements for the 1-tree Parks puzzles, the only two valid tree configurations for 1-tree $4 \times 4$ Parks puzzles}
        \label{fig:1shuriken}
    \end{figure}

    \begin{figure}[hbtp]
        \centering
        \includegraphics[height=0.5in]{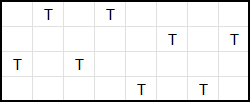}\vspace{12pt}\\
        \includegraphics[height=0.5in]{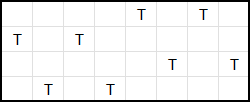}
        \caption{The Shuriken arrangements for the (1, 2)-tree Parks puzzles, the only two valid tree configurations for (1, 2)-tree $4 \times 8$ Parks puzzles.}
        \label{fig:12shuriken}
    \end{figure}

    The $(c,r)$-tree versions of the Shuriken arrangements are $4c \times 4r$ grids that simply extend the configurations in Figures \ref{fig:1shuriken} and \ref{fig:12shuriken} horizontally and/or vertically. To be specific, we look at the $2 \times 2$ box in each of the four corners of the central $4 \times 4$ square and simply replicate this type of box to the entire corresponding quadrant of the grid, so all the $2 \times 2$ boxes in any one quadrant of the grid will have their trees in the same place. For example, the box consisting of C3, D3, C4, D4 in the upper-left corner of the $4 \times 4$ middle square in Figure \ref{fig:8x8config} is of the same type (type C) as the other three boxes in the upper-left quadrant, namely the box A1, B1, A2, B2; the box C1, D1, C2, D2; and the box A3, B3, A4, B4. Similarly, the box in E5, F5, E6, F6 is the same type (type B) as the $2 \times 2$ boxes in the lower-right quadrant in Figure \ref{fig:8x8config}, and so on for the other two quadrant. 
    
    \begin{figure}[hbtp]
        \centering
        \includegraphics[height=1in]{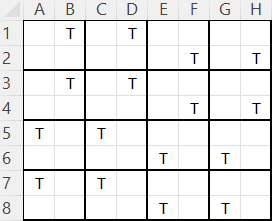}
        \caption{An $8 \times 8$ Shuriken arrangement, one of two valid tree configurations for a 2-tree $8 \times 8$ Parks puzzle}
        \label{fig:8x8config}
    \end{figure}
    
    \begin{figure}[hbtp]
        \centering
        \includegraphics[height=1in]{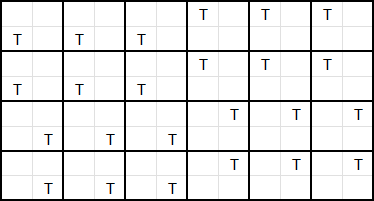}
        \caption{An $8 \times 12$ Shuriken arrangement, one of two valid tree configurations for a (2, 3)-tree $8 \times 12$ Parks puzzle}
        \label{fig:8x12config}
    \end{figure}

    To show that the two Shuriken arrangements are the only two valid tree configurations for the $(c,r)$-tree $4c \times 4r$ puzzle, we continue the reasoning introduced to show that $4c \times 4r$ is the minimal size and  break the grid up into $2 \times 2$ boxes as shown in Figures \ref{fig:8x8config} and \ref{fig:8x12config}. We need to put in $4cr$ trees ($c$ trees in each of the $4r$ columns), which is exactly equal to the number of boxes, $(2c)(2r)$, so by the pigeonhole principle we need to put exactly 1 tree in each $2 \times 2$ box.

    Note that if we put a tree into the right column of a $2 \times 2$ box, then any box to the right of it must have its tree in the right column (see the top half of Figure \ref{fig:8x8config} and the bottom half of Figure \ref{fig:8x12config}). Because the 2nd column must have $c$ trees in it, there are $c$ boxes in the first 2 columns of the grid that have the tree on the right side, and this forces all the boxes to the right of those to have the tree on the right as well, so at least $2rc$ boxes have trees on the right. Similarly, the boxes that have trees in the 2nd column from the right force the boxes left of them to have trees in the left columns (see bottom half of Figure \ref{fig:8x8config} and top half of Figure \ref{fig:8x12config}). Together, this forces exactly $2rc$ boxes to have a tree in the left column and $2rc$ boxes to have a tree in the right column. 
    
    The reasoning above implies that in a single row of $2r$ boxes, either all the boxes in that row have trees on the right or all of them have trees on the left. Because each of the two rows within a row of $2r$ boxes has $r$ trees, half of these boxes will have the tree in the top and half have the tree in the bottom. Thus we see that half of the rows of boxes ($c$, to be precise) contain boxes of type A or B from Figure \ref{fig:4boxes} (in equal amounts) and the other $c$ rows contain boxes of type C or D (again, in equal amounts). Applying this reasoning to columns, we see that $c$ columns of boxes have trees in the top of all the boxes (i.e. have boxes of type A or C, in equal amounts) and the other $c$ columns of boxes have trees in the bottom of all the boxes (i.e. have boxes of type B or D, in equal amounts).

    \begin{figure}[hbtp]
        \centering
        \includegraphics[height=0.75in]{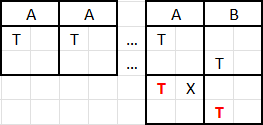}
        \caption{Junction where a horizontal series of Type A boxes switches to a Type B box (only the boxes in the top row are labelled for convenience)}
        \label{fig:JunctionAB}
    \end{figure}
    
    Suppose we go through a row of boxes from left to right and the row starts with a box of type A (see top row of Figure \ref{fig:JunctionAB}). Then eventually (at most halfway across) we will encounter a box of type B, and let's see what happens at this junction. The column of boxes where the box of type A is located only contains boxes of type A or C, so the box directly below the one of type A (if it exists) must be of type A (it can't be of type C because the tree in that one would be adjacent to the tree in the box of type B - note the X in Figure \ref{fig:JunctionAB}). That means the next row in Figure \ref{fig:JunctionAB} (if there is one) must also contain boxes of type A or B. Furthermore, the box below the one of type B must also be of type B, since it is in a row of Type A or B boxes and in a column of Type B or D boxes. Thus, the red \textbf{\color{red}T}'s in Figure \ref{fig:JunctionAB} are forced. Now the same reasoning applies to the next row, forcing any row below that one to have the same junction right below, and so on, meaning all of these $c$ rows of boxes of type A or B must be together, in the bottom of the grid. 

    \begin{figure}[hbtp]
        \centering
        \includegraphics[height=0.75in]{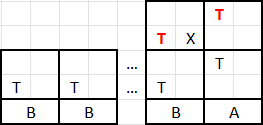}
        \caption{Junction where a horizontal series of Type B boxes switches to a Type A box (only the boxes in the bottom row are labelled for convenience)}
        \label{fig:JunctionBA}
    \end{figure}
    
    Conversely, as Figure \ref{fig:JunctionBA} shows, if a row of boxes of type A or B started with a box of type B on the left and eventually was followed by a box of type A, then at that junction, the box above the one of type B (if any) must be another of type B, and in this case all of these $c$ rows of boxes of type A or B must be together, in the \textit{top} half of the grid. It follows that the same is true of the rows that have the other two types of boxes (C and D) - they are just on the opposite side of the grid from the rows that have boxes of type A or B. Similarly, the columns that have boxes of type A or C must be together, either in the left or the right half of the grid, and this implies that in each row of boxes, the switch between A and B happens in the middle, so each quadrant of the grid must be made up of the same type of $2 \times 2$ box. 
    
    To summarize, if a Type A box exists in the left of a row of boxes, then all the rows of type A or B boxes are together on the bottom of the grid (this is shown in the top of Figure \ref{fig:12shuriken} and in Figure \ref{fig:8x8config}), and if a Type B box is on the left of a row, then the rows of type A or B are at the top of the grid (see the bottom of Figure \ref{fig:12shuriken} and also Figure \ref{fig:8x12config}), so the only two possible tree configurations for a $(c, r)$-tree $4c \times 4r$ puzzle are the Shuriken arrangements.
\end{proof}

\section{Number of Possible Tree Configurations}
\label{sec:NumTreeConfigs}
In this section, we make a few remarks about counting the total number of valid puzzles with fixed parameters that may be given to a human solver; in general, this is hard, and even the simpler problem of finding the number of valid tree configurations for puzzles is an open problem. We will not distinguish between two puzzles that are the same up to a permutation of their colors. There is no known fast algorithm for finding the number of $(c, r)$-tree puzzles with contiguous parks. 

We have been able to compute experimentally (via a SageMath program) that there are exactly 5880 1-tree $4 \times 4$ puzzles with \textit{contiguous} parks that have only the solution that is shown in the left (or equivalently right) side of Figure \ref{fig:1shuriken}, so there is a total of 11760 1-tree puzzles with unique solution. 

If the parks can be disconnected, i.e. \textit{non-contiguous}, then it is easy to count the $4 \times 4$ 1-tree puzzles with unique solution combinatorially. The puzzles that have only the left solution in Figure \ref{fig:1shuriken} may be assumed to have Park 0 at the tree in the top row, Park 1 at the tree in the next row, Park 2 at the next tree, and Park 3 at the last row. The total number of such puzzles is $4^{12}$ (since we can assign any number 0-3 to the rest of the squares). Now such puzzles have two solutions (i.e. also have the right side of Figure \ref{fig:1shuriken} as a solution) if and only if the park numbers assigned to the squares corresponding to trees in the right side of Figure \ref{fig:1shuriken} are all distinct. Thus there are $4^{12} - 4!4^8 = 15,204,352$ (possibly non-contiguous) puzzles that have only the left side of Figure \ref{fig:1shuriken} as a solution, for a total of $2\cdot 4^{12} - 4!4^8 = 31,981,568$ puzzles with at least 1 solution and $2(4^{12} - 4!4^8) = 30,408,704$ puzzles with exactly 1 solution. Note that only 11,760, or about 0.04\% of those, have contiguous parks.

The total number of $(c, r)$-tree puzzles (for any $c$ and any $r$) that may have 0, 1, or more than one solution and may have \textit{non-contiguous} parks is given by the Stirling numbers of the second kind $S(n, r) = {n \brace r}$. Recall that ${n \brace r}$ is the number of ways to partition a set of size $n$ into $r$ nonempty subsets, so ${mn \brace m}$ is precisely the total number of $m \times n$ puzzles. See, e.g. \cite{ConcreteMath} for more information on these numbers. The problem of counting the number of puzzles with at least one solution, or a unique solution, seems harder than counting all the possible tree arrangements that can occur in a Parks puzzle, so we now turn to counting the number of possible tree arrangements.



Table \ref{tab:1x1num_configs} below lists the number $c_1(n) = c_{(1, 1)}(n, n)$ of valid tree configurations for the 1-tree $n \times n$ puzzles where $4 \le n \le 13$. The sequence $\{c_1(n)\}$ is the {OEIS sequence A002464} \cite{oeisA002464} (the reader will note that the $4 \times 4$ case is already proven in Theorem \ref{thm:smallest_size}).

\begin{table}[hbtp]
    \centering
    \begin{tabular}{c|c}
    Size ($n \times n$) & $c_1(n)$\\
    $4 \times 4 $ & 2\\
    $5 \times 5 $ & 14\\
    $6 \times 6 $ & 90  \\
    $7 \times 7 $ & 646  \\
    $8 \times 8 $ & 5,242 \\
    $9 \times 9 $ & 47,622  \\
    $10\times 10$ &  479,306  \\
    $11\times 11$ &  5,296,790  \\
    $12\times 12$ &  63,779,034  \\
    $13 \times 13$ &  831,283,558 
    \end{tabular}
    \caption{Number of valid 1-tree $n \times n$ tree configurations}
    \label{tab:1x1num_configs}
\end{table}

As $n$ becomes larger, $c_1(n)$ grows on the order of $n!$. 
The problem of finding these numbers $c_1(n)$ dates back to at least 1887 \cite{Hertz}, when Severin Hertzsprung wrote a letter describing this sequence of numbers, originally thinking of it as either (a) listing the numbers from $1$ to $n$, never listing any two adjacent numbers next to each other, or (b) as the number of ways to arrange pieces on a chess board that act like a combination of a rook and a king, so that no piece can attack another. According to \cite{Riordan}, these numbers satisfy the recurrence relation: 
\begin{align}
    \label{rec_rel}
    c_1(n) &= (n+1)c_1(n-1) - (n-2)c_1(n-2)\\
    &- (n-5)c_1(n-3) + (n-3)c_1(n-4) \notag
\end{align} An explicit formula for $c_1(n)$ and much more general sequences is given in \cite{AbramsonMoser}:
\begin{equation}
    \label{eqn:explicit}
    c_1(n) = n! + \sum_{k=1}^{n-1} (-1)^k(n-k)!\sum_{r=1}^k 2^r{n-k \choose r}{k-1 \choose r-1}
\end{equation}
A more accessible explanation of \eqref{eqn:explicit} can be found in the educational video lecture by Another Roof \cite{lifelong_math_obsession}. See \cite{Bagno} for more information on the sequence $\{c_1(n)\}$ and \cite{Rukavicka} and \cite{Tauraso} for some interesting related applications.

In general, given $c$ and $r$, we can consider the sequence $\{c_{(c, r)}(m, n)\}$ which counts the number of all possible $m \times n$ size $(c, r)$-tree Parks puzzle tree configurations. More precisely, note that any valid size $m \times n$ of a $(c, r)$-tree puzzle has $n = rm/c$ as an integer, so $m$ must be divisible by $c' = c/\gcd(r, c)$. If we set $m = c'i$, then $n = r'i$ with $r' = r/\gcd(r, c)$, so the sequence becomes $\{c_{(c, r)}(c'i, r'i)\}_{i = 0}^{\infty}$. This is a doubly-infinite generalization of the $c = r = 1$ case, which gives the OEIS sequence A002464. Presumably, the correct generalization would have $c_{(c, r)}(0, 0) = 1$ since $c_1(0) = 1$ makes the recursion \eqref{rec_rel} and explicit formula \eqref{eqn:explicit} work; the terms $c_{(c, r)}(c'i, r'i)$ are 0 when $0 < c'i < 4c$ (i.e. $0 < i < \gcd(c, r)$); and $c_{(c, r)}(4c, 4r) = 2$, according to Theorem \ref{thm:smallest_size}.

Experimentally, we have been able to find the number $c_{(c, r)}(m, n)$ (with $n = rm/c$) for only very small $c$ and $r$ and only the first few possible sizes of puzzle larger than $4c \times 4r$, as shown in Table \ref{tab:num_configs}. With just our limited computations, we can see that for a given fixed $c$ and $r$, not both equal to 1, the corresponding sequence of possible tree configurations is currently not found in the OEIS; in fact, many of the numbers in this table are not in the OEIS at all. We hope to explore the asymptotics of the growth of $\{c_{(c, r)}(m, n)\}$ and/or find a recurrence relation similar to \eqref{rec_rel} or an explicit formula similar to \eqref{eqn:explicit} in future papers.

\begin{table}[hbtp]
    \centering
    \begin{tabular}{|c|c||c|c|}
    \hline
    \textbf{Size} & $c_{(2, 2)}(n, n)$ & \textbf{Size} & $c_{(1, 2)}(n, 2n)$\\
    \hline
    $9 \times 9$ & 664          & $5 \times 10$ & 282\\
    \hline
    $10 \times 10$ & 146,510    & $6 \times 12$ & 25,922\\
    \hline
    $11 \times 11$ & 31,197,434 & $7 \times 14$ & 2,755,942\\
    \hline
    $12 \times 12$ & 6,798,881,226 & $8 \times 16$ & 363,371,498\\
    \hline
    \textbf{Size} & $c_{(3, 3)}(n, n)$ & \textbf{Size} & $c_{1, 3}(n, 3n)$ \\
    \hline
    $13 \times 13$ & 42,732 & $5 \times 15$ & 7,912\\
    \hline
    $14 \times 14$ & 596,777,194 & $6 \times 18$ & 12,298,122  \\
    \hline
    \textbf{Size} & $c_{(4, 4)}(n, n)$ & $7 \times 21$ & 22,221,905,440\\
    \hline
    $17 \times 17$ & 3,758,020 & \textbf{Size} & $c_{(1, 4)}(n, 4n)$ \\
    \hline
    \textbf{Size} & $c_{(2, 3)}(2n, 3n)$ & $5 \times 20$ & 261,440\\
    \hline
    $10 \times 15$ & 151,930,062 & $6 \times 24$ & 7,373,997,162\\
    \hline
    \textbf{Size}  & $c_{(2, 4)}(n, 2n)$ & \textbf{Size} & $c_{(1, 5)}(n, 5n)$\\
    \hline
    $9 \times 18$ & 1,026,750 & $5 \times 25$ & 9,525,432\\
    \hline
    \end{tabular}
    \caption{Number of valid $(c, r)$-tree configurations for small $c$ and $r$}
    \label{tab:num_configs}
\end{table}  



\section{Conclusion}
We have shown that the $(c,r)$-tree Parks Puzzle is NP-complete for each $c$ and $r$, thus producing a doubly-infinite family of puzzles, each of which is NP-complete. We have also shown that there are interesting related sequences $\{c_{(c, r)}(m, n)\}$ associated with these puzzles, describing the number of possible tree arrangements for the possible sizes $m \times n$ of $(c, r)$-tree Parks puzzles. These sequences currently do not appear in the OEIS, and we have computed a few terms of the sequences. We hope to explore these sequences in more depth in a future paper; in particular, we hope to find recursive formulas such as \eqref{rec_rel} and/or explicit formulas such as \eqref{eqn:explicit}.

\section{Acknowledgments}
The inspiration for this work was the Logic Games app, developed by Andrea Sabbatini in 2011. We are grateful to Dr. Koorosh Firouzbakht for reading the manuscript and suggestions for improvement. We are also grateful to Dr. Fariba Khoshnasib-Zeinabad for many helpful discussions. 



\ifCLASSOPTIONcaptionsoff
  \newpage
\fi



%

\bibliographystyle{IEEEtran}
\bibliography{main_ieee}

%








\end{document}